\documentclass[11pt]{amsart}

\usepackage{txfonts}
\usepackage{amssymb}
\usepackage{graphics}
\usepackage{pst-node}
\usepackage{mathtools}
\usepackage{bbm}
\usepackage{array}
\usepackage{amsthm}

\newtheorem{theorem}{Theorem}[section]
\newtheorem{lemma}[theorem]{Lemma}
\newtheorem*{MT}{Main Theorem}

\theoremstyle{definition}

\theoremstyle{remark}

\numberwithin{equation}{section}

\newcommand{\comments}[1]{}

\begin{document}

\title{Complex-Time Singularity and Locality Estimates for Quantum Lattice Systems}

\author{Gabriel Bouch}

\address{Department of Mathematics, Rutgers University, Piscataway, New Jersey 08854}

\begin{abstract}
We present and prove a well-known locality bound for the complex-time dynamics
of a general class of one-dimensional quantum spin systems. Then we discuss how one might hope
to extend this same procedure to higher dimensions using ideas related to the Eden growth process
and lattice trees. Finally, we demonstrate with a specific family of lattice trees in the plane 
why this approach breaks down in dimensions greater than one and prove that there exist interactions for which the complex-time dynamics blows-up in finite imaginary time.
\end{abstract}

\maketitle

\section{Introduction}
When $\Lambda$ is a finite, connected subset of $\mathbb{Z}^{d}$, the time evolution of an observable $A$ is given by

\begin{equation}
 \tau^{\Lambda}_{t}(A)=e^{itH_{\Lambda}}Ae^{-itH_{\Lambda}}\, ,
\end{equation}
where $H_{\Lambda}$ is a bounded self-adjoint operator determined by a (typically) translation-invariant interaction. Clearly, $\tau^{\Lambda}_{t}(A)$ is an entire analytic function of $t$. If we take the limit $\Lambda \to \mathbb{Z}^{d}$, is the resulting function entire? If not, what is the nature of the singularities?\\

More than four decades ago, Araki \cite{Araki} established that, for a general class of one-dimensional interactions, the infinite volume time-evolution of a local observable is entire analytic in the time variable. In fact, Araki's work established a locality principle for the \textit{complex}-time evolution of local observables: the support of a local observable stays bounded (up to a small correction) as it evolves in complex-time no matter how large the system is. In \cite{LenciReyBellet} Lenci and Rey-Bellet use this work of Araki to obtain upper bounds on large deviations in quantum lattice systems. More recently, Ogata \cite{Ogata}, again relying heavily on Araki, has established a more complete large deviations principle for such systems. In a different application, T. Matsui, \cite{Matsui1} and \cite{Matsui2}, has used Araki's work to prove several central limit theorems for one-dimensional quantum spin systems. If Araki's results could be extended to dimensions greater than one, some of these applications would immediately generalize.\\

Coming from the opposite dimensional extreme, V. E. Zobov points out in \cite{Zobov} that he earlier \cite{Zobov2} ``established that the autocorrelation function of the Heisenberg magnet on an infinite-dimensional lattice at an infinite temperature has singular points on the imaginary time axis at a finite distance from the origin.'' But then he adds, \textit{``No rigorous results are known for systems of an arbitrary dimension $d$''} (emphasis added). So although something has been known about locality in dimension one and singularities in infinite dimensions for quite a while, general results in finite dimensions greater than one are completely lacking.\\

Araki's result for one-dimensional systems is similar to what are now referred to as Lieb-Robinson bounds \cite{LiebRobinson}. The primary difference is that in complex-time, the support of the observable may grow exponentially in the magnitude of the complex-time variable rather than linearly as is the case for the real-time dynamics. In a recent review of locality results for quantum spin systems, Nachtergaele and Sims \cite{NachSims} mention that it would be interesting if further progress could be made on complex-time locality results. After briefly discussing Araki's result, they make a comparison with the stochastic dynamics of some classical particle systems and hold out hope that this exponential growth in $|z|$ might be improved upon in at least some physically interesting systems (including, presumably, systems in dimensions greater than one).\\

In a novel application of a very simple complex-time Lieb-Robinson type bound, M. Hastings \cite{Hastings} recently gave quantitative bounds on a very interesting question concerning almost commuting Hermitian matrices. After his short proof of the Lieb-Robinson type bound he needs, Hastings remarks: "The proof of this Lieb-Robinson bound is significantly simpler than the proofs of the corresponding bounds for many-body systems considered elsewhere. The power series technique used here does not work for such systems." Although the power series technique does not give general locality results with supports growing \textit{linearly} (in the magnitude of the complex time), it \textit{does} work to prove Araki's result in one-dimension. And it is not at all obvious that the power series technique will \textit{not} work to prove similar results in higher dimensions.\\

In this work we shall demonstrate through the construction of a specific example that general complex-time locality results do \textit{not} hold in dimensions greater than one, and that the complex-time dynamics can blow-up in finite imaginary time. We have the following:

\begin{MT}\label{MainTheorem}
There exists a translation-invariant nearest-neighbor interaction on $\mathbb{Z}^{2}$, with the interaction between nearest-neighbor sites $H_{x_{1},x_{2}}$ satisfying $\left\Vert H_{x_{1},x_{2}} \right\Vert = 1$, an increasing sequence of square sublattices $\{ \Lambda_{j} \}$, and an observable $A$ supported at the origin such that
\begin{equation}
 \displaystyle{\lim_{j \to \infty}   \left\lVert  e^{izH_{\Lambda_{j}}}Ae^{-izH_{\Lambda_{j}}}   \right\rVert   = \infty}
\end{equation}
for $z$ purely imaginary and $|z| > 4^{21}$.
\end{MT}

\section{Quantum Spin Systems}
We consider quantum systems defined on finite subsets $\Lambda$ of $\mathbb{Z}^{d}$. To each site $x \in \Lambda$ we associate an $m$-dimensional Hilbert space $\mathcal{H}_{x}$. The Hilbert space of states is given by $\mathcal{H}_{\Lambda}=\bigotimes_{x \in \Lambda} \, \mathcal{H}_{x}$. For each site $x$, the observables are the complex $m \times m$ matrices, $M_{m}$. The algebra of observables for the whole system is $\mathcal{A}_{\Lambda}=\bigotimes_{x \in \Lambda} \, M_{m}$. If $X \subset \Lambda$, then, by identifying $A \in \mathcal{A}_{X}$ with $A \, \otimes \, \mathbbm{1} \in \mathcal{A}_{\Lambda}$, we have $\mathcal{A}_{X} \subset \mathcal{A}_{\Lambda}$. The support of an observable $A \in \mathcal{A}_{\Lambda}$ is the minimal set $X \subset \Lambda$ for which $A=A' \, \otimes \, \mathbbm{1}$ with $A' \in \mathcal{A}_{X}$. We can also consider the normed algebra $\cup_{\Lambda} \mathcal{A}_{\Lambda}$, where the union is taken over all finite subsets of $\mathbb{Z}^{d}$. We define the algebra of \textit{quasi-local} observables, $\mathcal{A}$, to be the norm completion of this normed algebra. If an element $A \in \mathcal{A}$ is in some $\mathcal{A}_{\Lambda}$, then we say that $A$ is a \textit{local} observable.
\\

An interaction $\Phi$ is a map from the finite subsets of $\mathbb{Z}^{d}$ to $\cup_{\Lambda} \mathcal{A}_{\Lambda}$ such that $\Phi (X) \in \mathcal{A}_{X}$ and $\Phi (X)=\Phi (X)^{*}$ for all finite $X \subset \mathbb{Z}^{d}$. The range of an interaction is defined to be the smallest $R>0$ such that $\Phi (X)=0$ when $\text{diam}(X) > R$. A quantum spin model is defined by a family of local Hamiltonians, parametrized by finite subsets $\Lambda \subset \mathbb{Z}^{d}$, given by

\begin{equation}
 H^{\Phi}_{\Lambda}=\sum_{X \subset \Lambda} \Phi(X).
\end{equation}
\\
We will consistently suppress the $\Phi$ in this notation. The complex-time evolution generated by a quantum spin model, $\{ \tau^{\Lambda}_{z}\}_{z \in \mathbb{C}}$, is defined by

\begin{equation}
 \tau^{\Lambda}_{z}(A)=e^{izH_{\Lambda}}Ae^{-izH_{\Lambda}}, \hspace{5 mm} A \in \mathcal{A}_{\Lambda}.
\end{equation}
\\
We will focus in this work on nearest-neighbor interactions. That is, $\Phi(X)=0$ whenever $X$ is \textit{not} of the form $X=\{x_{1},x_{2}\}$ where $\text{dist}(x_{1},x_{2})=1$. We will write $H_{x_{1},x_{2}}$ instead of $\Phi_{\{x_{1},x_{2}\}}$.\\

To begin analyzing $\tau^{\Lambda}_{z}(A)$, where $A$ is an observable supported at the origin, we consider the Taylor series
\begin{equation}\label{powerseries}
 \tau^{\Lambda}_{z}(A)=\sum_{n=0}^{\infty}\frac{z^{n}}{n!}\mathcal{C}_{\Lambda}^{n}(A)
\end{equation}
where $\mathcal{C}_{\Lambda}(A) \coloneqq [H_{\Lambda},A]$, the commutator of $H_{\Lambda}$ with $A$. Note that $\mathcal{C}_{\Lambda}^{n}(A)$ can be written as a sum of ``iterated commutators'' of the form $\left[H_{x_{n},y_{n}},\left[H_{x_{n-1},y_{n-1}},...,[H_{x_{1},y_{1}},A]...\right]\right]$ where:

\begin{enumerate}
 \item Either $x_{1}$ is the origin or $y_{1}$ is the origin.
 \item For  $2 \leq i \leq n$, at least one of $x_{i}$ or $y_{i}$ is in $\left\lbrace x_{1},y_{1},...,x_{i-1},y_{i-1} \right\rbrace $.
\end{enumerate}

\medskip
For example, in two dimensions,
\begin{multline}
\mathcal{C}_{\Lambda}^{1}(A) = \mathcal{C}_{\Lambda}(A) = [H_{(0,0), (1,0)},A] + [H_{(0,0), (0,1)},A]+ [H_{(0,0), (-1,0)},A] + [H_{(0,0), (0,-1)},A]
\end{multline}
independent of how large $\Lambda$ is since $A$ commutes with all of the other terms in $H_{\Lambda}$. Similarly,

\begin{multline}
\mathcal{C}_{\Lambda}^{2}(A) = \left[H_{(2,0), (1,0)},[H_{(0,0), (1,0)},A]\right] + \left[H_{(1,1),(1,0)},[H_{(0,0),(1,0)},A]\right]\\
+ \left[H_{(0,0),(1,0)},[H_{(0,0),(1,0)},A]\right] + \left[H_{(1,-1),(1,0)},[H_{(0,0),(1,0)},A]\right]\\
+ \left[H_{(0,0),(0,1)},[H_{(0,0),(1,0)},A]\right] + \left[H_{(-1,0),(0,0)},[H_{(0,0),(1,0)},A]\right]\\
+ \left[H_{(0,0),(0,-1)},[H_{(0,0),(1,0)},A]\right] + \ldots +\left[H_{(0,-2),(0,-1)},[H_{(0,0),(0,-1)},A]\right] \, .
\end{multline}

We would like to know how many sequences (which we shall call \textit{commutator sequences}) of nearest-neighbor pairs (equivalently \textit{edges}) $\{x_{1},y_{1}\},...,\{x_{n},y_{n}\}$ satisfying the above conditions there are. If the number of commutator sequences does not grow too fast in $n$, we might hope that $\tau^{\Lambda}_{z}(A)$ is well-approximated by an observable with ``small'' support, which would lead to a proof that the infinite volume limit of $\tau^{\Lambda}_{z}(A)$ exists for all $z$ and is entire analytic.\\

Let $X_{j}^{n}$ be the number of commutator sequences such that the set $\left\lbrace x_{1},y_{1}, \dots ,x_{n},y_{n} \right\rbrace $ contains exactly $j$ distinct points on the lattice. The corresponding connected collection of distinct lattice points is a \textit{lattice animal}. We shall denote it $L(\{x_{1},y_{1}\},...,\{x_{n},y_{n}\})$. Given that $|L(\{x_{1},y_{1}\},...,\{x_{n},y_{n}\})|=j$, the number of nearest-neighbor pairs (equivalently edges) that can be formed from elements of $L(\{x_{1},y_{1}\},...,\{x_{n},y_{n}\})$ is bounded above by $dj$. (To see this, just associate to each lattice point in $L(\{x_{1},y_{1}\},...,\{x_{n},y_{n}\})$ the $d$ distinct neighbors obtained by increasing one of the coordinates by one.)\\

To any lattice animal $L$ we can also associate the \textit{perimeter} $p(L)$ which we define to be the collection of nearest-neighbor pairs $\{x,y\}$ such that $x \in L$ and $y \notin L$. We call a commutator sequence $\{x_{1},y_{1}\},...,\{x_{n},y_{n}\}$ a \textit{lattice animal history of length $n$}, or just \textit{history of length $n$}, if, in addition to being a commutator sequence, it satisfies the following condition: For $2 \leq i \leq n$, $\{x_{i},y_{i}\} \in p(L(\{x_{1},y_{1}\},...,\{x_{i-1},y_{i-1}\}))$. Finally, let $\bar{p}_{n}$ be the average perimeter over all lattice animal histories of length $n-1$. We define the lattice animal associated with the history of length $0$ (the ``empty" history) to be $L_{0} \coloneqq \{ 0 \}$ so that we have $\bar{p}_{1}=2d$.\\

As an example, consider the following commutator sequence :
\begin{equation*}
 \{(0,0),(1,0)\}, \{(1,0),(1,1)\}, \{(1,1),(2,1)\}, \{(1,0),(1,1)\}, \{(0,1),(1,1)\}, \{(0,0),(0,1)\} \, .
\end{equation*}
Figure~\ref{latticeanimal1} gives a graphical representation of this commutator sequence. The lattice point with the open circle is the origin. The edges associated with nearest-neighbor pairs in the sequence are indicated with solid lines and are labelled according to the positions in the commutator sequence in which they appear. The dotted edges represent perimeter edges associated with the underlying lattice animal determined by the commutator sequence. Note that the commutator sequence would be a lattice animal history of length $4$ if the fourth and sixth members of the sequence were deleted.\\

\begin{figure}
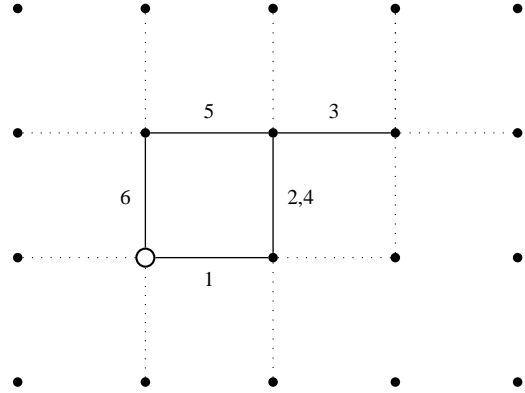

\begin{center}
$
\hspace{0.75cm}
 \begin{psmatrix}[emnode=dot,colsep=1.5cm,rowsep=1.2cm]
	& & & & &\\
	& & & & &\\
	& [mnode=circle] & & & &\\
	& & & & &
  \everypsbox{\scriptstyle}
  \ncline[linewidth=.5pt,linestyle=dotted]{1,2}{2,2}
  \ncline[linewidth=.5pt,linestyle=dotted]{1,3}{2,3}
  \ncline[linewidth=.5pt,linestyle=dotted]{1,4}{2,4}
  \ncline[linewidth=.5pt,linestyle=dotted]{1,4}{2,4}
  \ncline[linewidth=.5pt,linestyle=dotted]{2,1}{2,2}
  \ncline[linewidth=.5pt]{2,2}{2,3}^{5}
  \ncline[linewidth=.5pt]{2,3}{2,4}^{3}
  \ncline[linewidth=.5pt,linestyle=dotted]{2,4}{2,5}
  \ncline[linewidth=.5pt]{2,2}{3,2}<{6}
  \ncline[linewidth=.5pt]{2,3}{3,3}>{2,4}
  \ncline[linewidth=.5pt,linestyle=dotted]{2,4}{3,4}
  \ncline[linewidth=.5pt,linestyle=dotted]{3,1}{3,2}
  \ncline[linewidth=.5pt]{3,2}{3,3}_{1}
  \ncline[linewidth=.5pt,linestyle=dotted]{3,3}{3,4}
  \ncline[linewidth=.5pt,linestyle=dotted]{3,2}{4,2}
  \ncline[linewidth=.5pt,linestyle=dotted]{3,3}{4,3}
 \end{psmatrix}
$
\caption{A Commutator Sequence and the Perimeter of the Underlying Lattice Animal}
\label{latticeanimal1}
\end{center}
\end{figure}

To get an upper bound on $X_{j}^{n}$, we enumerate a slightly different collection of sequences. We begin by making the following definition.

\begin{equation}
 E^{\text{int}}(L) \coloneqq \{ e=\{ v_{1},v_{2} \} \hspace{0.08in} \text{an edge} \, | \, v_{1},v_{2} \in L \}
\end{equation}
Suppose $e_{1},e_{2}, \ldots ,e_{n}$ is a commutator sequence of length $n \geq 1$. (We will freely interchange edges and nearest-neighbor pairs without further comment.) It is helpful to consider the corresponding sequence of lattice animals, $L_{0},L(e_{1}),L(e_{1},e_{2}), \ldots ,L(e_{1},e_{2}, \ldots ,e_{n})$. If $|L(e_{1},e_{2}, \ldots ,e_{n})|=j \geq 2$, then for exactly $j-2$ indices, say $1<i_{2}<\ldots<i_{j-1} \leq n$, we must have $e_{i_{k}} \in P(L(e_{1}, \dots ,e_{i_{k}-1}))$. (We will always have $e_{1} \in P(L_{0})$.) So, $e_{1},e_{i_{2}},e_{i_{3}}, \ldots ,e_{i_{j-1}}$ is a lattice animal history of length $j-1$. Suppose $l \notin \{1,i_{2}, \ldots ,i_{j-1} \}$. Then, if $r$ is the greatest integer such that $i_{r} < l$, then $e_{l} \in E^{\text{in}}(L(e_{1},e_{i_{2}}, \ldots ,e_{i_{r}}))$.\\

Thus, a commutator sequence that will count toward $X_{j}^{n}$ consists of a sequence of $n$ edges in which exactly $j-1$ edges, including the first one, are perimeter edges of the ``immediately preceding lattice animal'', and $n-(j-1)$ edges are interior edges of the ``immediately preceding lattice animal''.\\

Let $Z_{j}^{n}$ be the number of sequences of length $n$ constructed in the following manner. First choose a lattice animal history of length $j-1$. Then choose $j-2$ positions other than the first position in the sequence of length $n$ in which to place the edges from the lattice animal history. Insert the edges from the lattice animal history, in order, into the first position and the additional $j-2$ positions in the sequence. Fill the remaining $n-(j-1)$ positions in the sequence with any of the first $dj$ members (allowing repeats) of an infinite collection of ``dummy edges'', say $\{f_{1},f_{2},f_{3}, \ldots \}$.\\

It is easy to see that $Z_{j}^{n}$ is greater than or equal to $X_{j}^{n}$. All commutator sequences that would be counted in $X_{j}^{n}$ can be constructed in a very similar fashion as the sequences counted in $Z_{j}^{n}$. The only difference is that commutator sequences have interior edges chosen from the immediately preceding lattice animal in the positions where the sequences in $Z_{j}^{n}$ have dummy edges. A lattice animal history $e_{1},e_{i_{2}},\ldots,e_{i_{j-1}}$ determines a unique collection of edges, $E^{\text{int}}(L(e_{1},e_{i_{2}},\ldots,e_{i_{j-1}}))$, from which these interior edges can be selected. (Actually, at most of the steps, the collection of edges from which an interior edge can be selected is strictly contained in this collection.) Since $|E^{\text{int}}(L(e_{1},e_{i_{2}},\ldots,e_{i_{j-1}}))| < dj$, $X_{j}^{n} \leq Z_{j}^{n}$.\\

We have the following.

\begin{lemma}
 For $n \geq 1$ and $j \geq 2$,
\begin{equation}\label{Xnjboundary}
 X^{n}_{n+1}=\prod_{i=1}^{n} \bar{p}_{i} \, , \hspace{0.1in} X^{n}_{j}=0 \hspace{0.08in} \text{if $j > n+1$} \, ,
\end{equation}
and
\begin{equation}\label{XnjZnj}
 X^{n}_{j} \leq Z^{n}_{j} \leq (2d)^{n-1}\bar{p}_{1} \cdot \ldots \cdot \bar{p}_{j-1} j^{n-(j-1)} \, .
\end{equation}

\end{lemma}

\begin{proof}
 Obviously $X^{n}_{j}=0$ if $j>n+1$ since in this case a lattice animal history of length $j-1$ has more than $n$ edges. The other equation in \eqref{Xnjboundary} results from the definition of $X^{n}_{n+1}$ and $\bar{p}_{n}$. We note that $X^{n}_{n+1}$ is precisely the number of lattice animal histories of length $n$. Since all lattice animal histories of length $n$ result from adding a perimeter edge to a lattice animal history of length $n-1$, we have

\begin{equation}
 X^{n}_{n+1}=\sum_{H_{n-1}} |p(L(e_{1},\ldots,e_{n-1}))|
\end{equation}
where $H_{n-1}$ is the collection of lattice animal histories of length $n-1$. In addition, by definition,
\begin{equation}
 \bar{p}_{n}=\frac{\sum_{H_{n-1}} |p(L(e_{1},\ldots,e_{n-1}))|}{X^{n-1}_{n}} \, .
\end{equation}
So, $X^{n}_{n+1}=\bar{p}_{n}X^{n-1}_{n}$. Also, $X^{1}_{2}=\bar{p}_{1}$ since the number of lattice animal histories of length one is just the number of edges containing the origin, which is $2d=\bar{p}_{1}$.\\

To understand \eqref{XnjZnj}, we consider again how the sequences counting toward $Z^{n}_{j}$ are constructed. First, a lattice animal history of length $j-1$ is chosen. By \eqref{Xnjboundary}, there exist exactly $\bar{p}_{1} \cdot \ldots \cdot \bar{p}_{j-1}$ such sequences. Then, $j-2$ positions in the sequence of length $n$ in addition to the first position are selected. This can be done in $\binom{n-1}{j-2} \leq 2^{n-1}$ ways. Finally, each of the remaining $n-(j-1)$ positions are filled with dummy edges chosen from among a set of size $dj$. This can be done in $(dj)^{n-(j-1)}$ ways. Thus,

\begin{equation}
\begin{split}
 X^{n}_{j} \leq Z^{n}_{j} &\leq \bar{p}_{1} \cdot \ldots \cdot \bar{p}_{j-1} 2^{n-1} (dj)^{n-(j-1)}\\
 &\leq (2d)^{n-1} \bar{p}_{1} \cdot \ldots \cdot \bar{p}_{j-1} j^{n-(j-1)} \, .
\end{split}
\end{equation}

\end{proof}

If it is true that

\begin{equation}\label{perbd}
\bar{p}_{j} \leq C_{1}\cdot j^{\alpha}
\end{equation}
for some constant $C_{1}$ depending only on $d$ and some nonnegative $\alpha < 1$, then we would have

\begin{equation}
\sum_{j=2}^{n+1} X_{j}^{n} \leq (2dC_{1})^{n-1} \displaystyle{\sum_{j=2}^{n+1} [ (j-1)! ]^{\alpha}j^{n-j+1}} \, .
\end{equation}
Stirling's approximation gives

\begin{equation}
 (j-1)! < \sqrt{2\pi(j-1)}\left(\frac{j-1}{e} \right)^{j-1}e = \frac{\sqrt{2\pi (j-1)}}{e^{j-2}} (j-1)^{j-1} \, .
\end{equation}
Therefore,

\begin{equation}
 \sum_{j=2}^{n+1} X_{j}^{n} < (2dC_{1})^{n-1} \displaystyle{\sum_{j=2}^{n+1} \frac{\left[ 2\pi (j-1) \right]^{\frac{\alpha}{2}}}{e^{\alpha(j-2)}} j^{n-\beta(j-1)}}
\end{equation}
where $\beta = 1-\alpha$. Therefore,

\begin{equation}\label{xnjfirstbound}
 \displaystyle{\sum_{j=2}^{n+1} X_{j}^{n} < (2dC_{1})^{n-1} n(2\pi n)^{\frac{\alpha}{2}} \max_{j \in \left\lbrace  2, \ldots ,n+1 \right\rbrace } j^{n-\beta(j-1)}} \, .
\end{equation}

This inequality gives an upper bound on the number of terms in $\mathcal{C}_{\Lambda}^{n}(A)$, where each term is of the form $\left[H_{x_{n},x_{n}+1},\left[H_{x_{n-1},x_{n-1}+1},...,[H_{x_{1},x_{1}+1},A]...\right]\right]$. A bound on the nearest-neighbor interaction implies a bound on the norm of this ``$n$-fold commutator'', which together with \eqref{xnjfirstbound} may enable us to show that \eqref{powerseries} converges absolutely. This would lead to a general locality result.

\section{One-Dimensional Systems}

In one dimension, we clearly have an estimate of the form \eqref{perbd}. In fact, we have $\bar{p}_{j} = 2$ for all $j$. Thus, 

\begin{equation}\label{Xnj1d}
 \displaystyle{\sum_{j=2}^{n+1} X_{j}^{n} < 4^{n-1} n \max_{j \in \left\lbrace  2, \ldots ,n+1 \right\rbrace } j^{n-j+1} \leq 8^{n-1} \max_{j \in \left\lbrace  2, \ldots ,n+1 \right\rbrace } j^{n-j+1}} \, .
\end{equation}

\begin{theorem}
 Fix a finite subset $\Lambda_{L}=\left\lbrace -L,-(L-1),...,L-1,L \right\rbrace   \subset \mathbb{Z}$ and suppose the interaction satisfies $  \left\lVert  H_{x_{1},x_{2}}  \right\rVert   \leq M$ for all nearest-neighbor pairs $\{x_{1},x_{2}\}$ in $\mathbb{Z}$. Then, there exist constants $C_{1}$ and $C_{2}$, not depending on $L$, such that for any observable $A$ and any positive integer $ m > \exp(C_{1}M|z|)$

\begin{equation}
 \tau^{\Lambda_{L}}_{z}(A)=B^{\Lambda_{L}}_{m, z} + C^{\Lambda_{L}}_{m, z}
\end{equation}
\\
where the support of $B^{\Lambda_{L}}_{m, z}$ stays within a distance $m$ of the support of $A$ and

\begin{equation}
   \left\lVert  C^{\Lambda_{L}}_{m, z}  \right\rVert   \leq | \text{spt}(A) |   \left\lVert  A  \right\rVert  C_{2}e^{-m}.
\end{equation}
If dist$(\text{spt}(A),\{ -L, L \}) > m$, then $B^{\Lambda_{L}}_{m, z}$ is independent of $L$.
\end{theorem}

\begin{proof} 
We want to investigate $\max_{j\in \{2,...,n+1\}}j^{n-j+1}$. Define

\begin{equation}
 f(x) \coloneqq (n-x+1)\ln x \, .
\end{equation}

Straightforward calculus shows that the maximum of this function occurs when $x_{0}$ satisfies

\begin{equation}\label{nx}
 n+1=x_{0}(1+\ln x_{0}) \, .
\end{equation}

Obviously $x_{0}^{n-x_{0}+1} \leq x_{0}^{n+1}$ and \eqref{nx} gives $x_{0}=\frac{n+1}{\ln x_{0} +1}$. So

\begin{equation}
 \begin{split}
  x_{0}^{n+1}& =\displaystyle{\left( \frac{n+1}{\ln x_{0} + 1}\right) ^{n+1}}\\
	& =\frac{(n+1)^{n+1}}{(\ln x_{0} + 1)^{n+1}} \, .
 \end{split}
\end{equation}

Since $2 \leq x_{0} \leq n+1$, \eqref{nx} also gives

\begin{equation}
 \ln x_{0} +1=\ln(n+1)-\ln(1+\ln x_{0})+1 \geq \ln(n+1)-\ln(1+\ln (n+1))+1 \, .
\end{equation}

In addition, it is easy to see that

\begin{equation}
 \ln(1+\ln (n+1)) \leq \frac{9}{10} \ln(n+1) \, .
\end{equation}

Therefore,

\begin{equation}
 \ln x_{0}+1 \geq \frac{1}{10}\ln(n+1)+1 >  \frac{1}{10}\ln(n+1) \, .
\end{equation}

Thus,

\begin{equation}
 x_{0}^{n+1} < \frac{10^{n+1}(n+1)^{n+1}}{(\ln(n+1))^{n+1}} \, .
\end{equation}

Inserting this in \eqref{Xnj1d} gives

\begin{equation}\label{ubsumxjn}
\begin{split}
 \sum_{j=2}^{n+1}X_{j}^{n}& < \displaystyle{8^{n-1} \frac{10^{n+1}(n+1)^{n+1}}{[\ln (n+1)]^{n+1}}}\\
 	& \leq \displaystyle{\frac{1}{64} \left(\frac{80}{\ln (n+1)} \right)^{n+1} n \left( \frac{n+1}{n} \right)^{n+1} n^{n}}\\
	& < \displaystyle{\frac{1}{64} \left( \frac{240}{\ln (n+1)} \right)^{n+1} n^{n}}\\
	& < \displaystyle{\frac{15}{2} \left( \frac{240}{\ln (n+1)} \right)^{n} n^{n}}\\
 	& < \displaystyle{\frac{15}{2} \left( \frac{240e}{\ln (n+1)} \right)^{n} n!} 
\end{split}
\end{equation}
where we have used the simple relations $\frac{n+1}{n} \leq \frac{3}{2}$ when $n \geq 2$, $n < 2^{n+1}$, $\ln 2 > \frac{1}{2}$, and $n^{n} < e^{n} n!$.\\

So, when spt$(A)=\{0\}$, \eqref{ubsumxjn} gives an upper bound on the number of terms in $\mathcal{C}_{\Lambda}^{n}(A)$, where each term is of the form $[H_{x_{n},x_{n}+1},[H_{x_{n-1},x_{n-1}+1},...,[H_{x_{1},x_{1}+1},A]...]]$. If we expand such an iterated commutator, it will have $2^{n}$ summands, each being a product of $n+1$ operators. Thus,

\begin{equation}
   \left\lVert  [H_{x_{n},x_{n}+1},[H_{x_{n-1},x_{n-1}+1},...,[H_{x_{1},x_{1}+1},A]...]]   \right\rVert   \leq 2^{n}M^{n}   \left\lVert  A  \right\rVert   \, .
\end{equation}

Therefore,

\begin{equation}
   \left\lVert \mathcal{C}_{\Lambda}^{n}(A)  \right\rVert  < \displaystyle{\frac{ 15 \left\Vert A \right\Vert}{2} \left( \frac{480eM}{\ln (n+1)} \right)^{n} n!} \, .
\end{equation}
\\

We also note that for general $A$ with finite support,

\begin{equation}
   \left\lVert \mathcal{C}_{\Lambda}^{n}(A)  \right\rVert  < \displaystyle{\frac{ 15 \left\Vert A \right\Vert \cdot |\text{spt}(A)|}{2} \left( \frac{480eM}{\ln (n+1)} \right)^{n} n!} \, .
\end{equation}

Now we choose m large enough that

\begin{equation}
 \frac{480eM|z|}{\ln (m+1)} < \frac{1}{e} \, .
\end{equation}

That is,
\begin{equation}
 m > \exp(480e^{2}M|z|)-1 \, .
\end{equation}

Then,

\begin{equation}
 \begin{split}
  \sum_{n=m}^{\infty}\frac{|z|^{n}}{n!}  \left\lVert \mathcal{C}_{\Lambda}^{n}(A)  \right\rVert  & < \sum_{n=m}^{\infty}\displaystyle{\frac{15 \left\Vert A \right\Vert \cdot |\text{spt}(A)|}{2} \left( \frac{1}{e} \right)^{n}}\\
	& = \displaystyle{\frac{15 \left\Vert A \right\Vert \cdot |\text{spt}(A)|}{2} \frac{e}{e-1} e^{-m}}\\
	& < \displaystyle{\frac{15}{2} \left\Vert A \right\Vert \cdot |\text{spt}(A)| e^{-m}} \, .
 \end{split}
\end{equation}

The theorem follows.
\end{proof}

This theorem enables us to define infinite volume complex-time-evolved local observables, and show that they are entire analytic functions of the time variable.

\begin{lemma}\label{UniformCauchy}
If $R > 0$ and $A$ is a local observable, then $\left\lbrace \tau^{\Lambda_{L}}_{z}(A) \right\rbrace_{L=L_{0}}^{\infty}$ (where $L_{0}$ is large enough that \text{spt}$(A)$ is contained in $\Lambda_{L_{0}}$) is uniformly Cauchy in $\overline{B(0,R)}$, the closed ball of radius $R$ centered at the origin.
\end{lemma}

\begin{proof}
 Let $\epsilon > 0$. Choose $m > \exp(C_{1}MR)$ such that $|\text{spt}(A)| \cdot \left\Vert A \right\Vert C_{2} e^{-m} < \frac{\epsilon}{2}$. Then choose $L_{1}$ large enough that dist(spt($A$), $\left\lbrace -L_{1}, L_{1} \right\rbrace$) $> m$. Suppose $L_{2}, L_{3} > L_{1}$. Then, for $|z| < R$,

\begin{equation}
\begin{split}
 \left\Vert \tau^{\Lambda_{L_{2}}}_{z}(A) - \tau^{\Lambda_{L_{3}}}_{z}(A) \right\Vert &= \left\Vert B^{\Lambda_{L_{2}}}_{m, z} + C^{\Lambda_{L_{2}}}_{m, z} - B^{\Lambda_{L_{3}}}_{m, z} - C^{\Lambda_{L_{3}}}_{m, z} \right\Vert \\
	& = \left\Vert C^{\Lambda_{L_{2}}}_{m, z} - C^{\Lambda_{L_{3}}}_{m, z} \right\Vert \\
	& \leq \left\Vert C^{\Lambda_{L_{2}}}_{m, z} \right\Vert + \left\Vert C^{\Lambda_{L_{3}}}_{m, z} \right\Vert \\
	& < \epsilon \, .
\end{split}
\end{equation}

\end{proof}

This lemma enables us to define the following infinite volume observable:

\begin{equation}
 \tau_{z}(A) \coloneqq \lim_{L \to \infty} \tau^{\Lambda_{L}}_{z}(A) \, .
\end{equation}

\begin{lemma}
 $\tau_{z}(A)$ is a continuous function of $z$.
\end{lemma}

\begin{proof}
 Fix $z_{0} \in \mathbb{C}$ and $R > |z_{0}|$. Let $\epsilon > 0$. Choose $m > \exp(C_{1}MR)$ such that $|\text{spt}(A)| \cdot \left\Vert A \right\Vert C_{2} e^{-m} < \frac{\epsilon}{3}$. Then choose $L$ large enough that dist(spt($A$), $\left\lbrace -L, L \right\rbrace$) $> m$ and $\left\Vert \tau_{z}(A) - \tau^{\Lambda_{L}}_{z}(A) \right\Vert < \frac{\epsilon}{3}$ for all $z$ satisfying $|z| < R$. Choose $\delta$ small enough that $\delta + |z_{0}| < R$ and $\left\Vert \tau^{\Lambda_{L}}_{z_{1}}(A) - \tau^{\Lambda_{L}}_{z_{0}}(A) \right\Vert < \frac{\epsilon}{3}$ if $|z_{1}-z_{0}| < \delta$. Then, for $|z_{1} - z_{0}| < \delta$,

\begin{equation}
 \begin{split}
  \left\Vert \tau_{z_{1}}(A) - \tau_{z_{0}}(A) \right\Vert & \leq \left\Vert \tau_{z_{1}}(A) - \tau^{\Lambda_{L}}_{z_{1}}(A) \right\Vert + \left\Vert \tau^{\Lambda_{L}}_{z_{1}}(A) - \tau^{\Lambda_{L}}_{z_{0}}(A) \right\Vert + \left\Vert \tau^{\Lambda_{L}}_{z_{0}}(A) - \tau_{z_{0}}(A) \right\Vert \\
	& < \epsilon \, .
\end{split}
\end{equation}

\end{proof}

We arrive at the following.

\begin{theorem}
 $\tau_{z}(A)$ is an entire analytic function.
\end{theorem}

\begin{proof}
 We will show that for any triangular path $T$ in the plane,

\begin{equation}
 \int_{T} \tau_{z}(A) \, dz = 0 \, .
\end{equation}
Let $\epsilon > 0$. Choose $R$ large enough that $T$ is contained in the ball of radius $R$ centered at the origin. Choose $m > \exp(C_{1}MR)$ such that $|\text{spt}(A)| \cdot \left\Vert A \right\Vert C_{2} e^{-m} < \frac{\epsilon}{\text{length}(T)}$ and $L$ as in the proof above. Then,

\begin{equation}
\begin{split}
 \left\Vert \int_{T} \tau_{z}(A) \, dz \right\Vert & \leq \left\Vert \int_{T} \left( \tau_{z}(A) - \tau^{\Lambda_{L}}_{z}(A) \right) \, dz + \int_{T} \tau^{\Lambda_{L}}_{z}(A) \, dz \right\Vert \\
	& \leq \int_{T} \left\Vert \tau_{z}(A) - \tau^{\Lambda_{L}}_{z}(A) \right\Vert \, dz + \left\Vert \int_{T} \tau^{\Lambda_{L}}_{z}(A) \, dz \right\Vert \\
	& \leq \epsilon \, .
\end{split}
\end{equation}

By Morera's Theorem, $\tau_{z}(A)$ is entire analytic.

\end{proof}

\section{Dimensions Greater Than One and the Eden Growth Model}
If a bound of the form \eqref{perbd} holds in dimensions greater than one, then a completely analogous argument would give the same result in higher dimensions. Is there any reason to hope that such a bound on $\bar{p}_{j}$ is true? Our procedure for constructing lattice animal histories is very similar to a discrete-step Markov process first considered by Murray Eden in \cite{Eden} and now referred to as an \textit{Eden growth process}. In an Eden growth process on $\mathbb{Z}^{d}$, the state space at step $n$ is the collection of all lattice animals of size $n+1$ containing the origin. Given a lattice animal $L(n-1)$ containing $n$ lattice points, the possible lattice animals at step $n$ are those which can be realized by adding a lattice point from the perimeter of $L(n-1)$. (Here a lattice point $y$ is on the perimeter of $L(n-1)$ if the nearest-neighbor pair $\{ x,y \}$ is a perimeter edge (as defined previously) for some $x \in L(n-1)$.) The probability of choosing any particular lattice point $y$ on the perimeter is 

\begin{equation}
\frac{\text{number of perimeter edges containing $y$}}{p(L(n-1))} \, .
\end{equation}
\\

Computer simulations of an Eden growth process demonstrate that the typical lattice animal containing a large number of lattice points grown by such a method is very nearly a ball in dimension $2$ \cite{Stauffer}. Establishing a rigorous upper bound on the expected perimeter in an Eden growth process is a very interesting problem in its own right. Using results of Kesten \cite{Kesten} on first-passage percolation and a method of Richardson \cite{Richardson} for associating an Eden growth process with a continuous-time process, it is straightforward to obtain the following result on the expected perimeter \cite{Bouch}.

\begin{theorem}
 The expected perimeter in a $d$-dimensional Eden growth process is bounded above by $Kn^{1-\frac{1}{d(2d+5) +1}}$ for some constant $K$.
\end{theorem}

Will this then lead to a locality result? Unfortunately, no. In an Eden growth process, not all histories of length $n$ are equally probable. So the expected perimeter in an Eden growth process is not necessarily the same as the average perimeter over all histories. This latter average perimeter seems particularly difficult to get a handle on, though one might hope that the above result for Eden growth processes is valid for the average over all histories as well. Theorem~\ref{MainTheorem}, however, implies the following.

\begin{theorem}
 No bound of the form $\bar{p}_{j} \leq C_{1}\cdot j^{\alpha}$ with $\alpha < 1$ and $C_{1}$ independent of $j$ exists for the average perimeter taken over all lattice animal histories of length $j-1$.
\end{theorem}

\section{Lattice Trees}

Obtaining a useful upper bound on the number of terms in $\mathcal{C}_{\Lambda}^{n}(A)$ by analyzing the average perimeter over all lattice animal histories appears to be a difficult way to move forward. In an Eden growth process one can make connections with continuous time processes and some machinery is available to work with. But such tools don't seem to be available for the average we are interested in. So in this section we present an alternate approach to counting the number of terms in the $n$-fold commutator $\mathcal{C}_{\Lambda}^{n}(A)$.\\

As we have noted already, each nearest-neighbor pair of lattice points in $\mathbb{Z}^{d}$ defines an edge of the cubical lattice in $\mathbb{Z}^{d}$. Thus, each commutator sequence defines a subgraph of the cubical graph in $\mathbb{Z}^{d}$. Each of these subgraphs contains at least one maximal spanning tree. It is known that the number of lattice trees with $k$ edges and containing the origin is bounded above by a function of the form $C^{k}$ for some constant $C$ depending only on the dimension $d$ \cite{Madras}. The potential utility of this fact can be seen from the following analysis.\\

Given a lattice tree with $k$ edges and containing the origin, how many orderings of the edges correspond to lattice animal histories of length $k$? In other words, given a fixed lattice tree with $k$ edges and rooted at the origin, how many ways can we ``construct'' this tree one edge at a time such that the first edge contains the origin and at each step we always have a tree?\\

Suppose that for all lattice trees with $k$ edges and rooted at the origin the number of ways of ``constructing'' such a tree is bounded above by

\begin{equation}\label{treebnd}
 \frac{(C_{1})^{k}k!}{(\ln k)^{k}} \, .
\end{equation}

Given a fixed lattice tree $T$ with $k$ edges and rooted at the origin, how many commutator sequences of length $n$ such that the induced subgraph is spanned by $T$ are there? It is relatively straightforward to obtain an upper bound. 
\\

In the commutator sequence we have at most $\binom{n}{k}$ choices for the positions of the edges (that is, nearest-neighbor pairs) associated with the spanning tree $T$. For each of these choices of $k$ locations, the edges in the spanning tree can be ordered in at most $\frac{(C_{1})^{k}k!}{(\ln k)^{k}}$ different ways. The remaining $n-k$ locations in the commutator sequence must be filled with edges whose associated vertices are contained in $T$. As indicated earlier, there are no more than $kd$ such nearest-neighbor pairs. Therefore we obtain an upper bound of

\begin{equation}
 \frac{(C_{1})^{k}k!}{(\ln k)^{k}}\binom{n}{k}(kd)^{n-k} \, .
\end{equation}
\\

The maximum value of this expression as $k$ ranges from $2$ to $n$ can be found using simple calculus as we did in the proof of the one-dimensional case. In this case we find that the maximum value is bounded above by

\begin{equation}
 \frac{(C_{2})^{n}n!}{[\ln(\ln n)]^{n}} \, .
\end{equation}

Since every commutator sequence of length $n$ can be associated with an underlying spanning tree containing at most $n$ edges, and since the number of rooted lattice trees grows at most exponentially, a bound of the form \eqref{treebnd} would lead to the result that the total number of commutator sequences of length $n$ is no more than

\begin{equation}
 \frac{(C_{3})^{n}n!}{[\ln(\ln n)]^{n}} \, .
\end{equation}

This would be strong enough to give us a slightly modified version of the locality result we found for one-dimensional systems.

\section{Counterexample}
In this section we will demonstrate that no such bound of the form \eqref{treebnd} exists for lattice trees in dimensions greater than one.

\begin{theorem}\label{Treethm}
 There exists an increasing sequence of positive integers $(n_{1},n_{2},n_{3}, \ldots)$ and a constant $C$ such that for each $n_{j}$ there exists a $2$-dimensional rooted lattice tree $T_{j}$ containing $n_{j}$ edges and having the property that the number of ways of constructing $T_{j}$ is greater than
\begin{equation}
 \frac{n_{j}!}{C^{n_{j}}} \, .
\end{equation}

\end{theorem}

We begin with a definition and a lemma. Given a rooted tree, a \textit{descendant} of an edge $e$ is any edge that can be reached after $e$ if one is traveling ``away'' from the root. We define the \textit{weight} of an edge $e$, $w(e)$, to be one plus the number of descendants of $e$. Then, we have the following lemma (communicated to me by Elizabeth Kupin).

\begin{lemma}
 Let $\{e_{1},e_{2},\ldots,e_{n}\}$ be the edges in a rooted tree $T$. The number of ways of constructing $T$ is given by
\begin{equation}\label{treeorders}
 \displaystyle{\frac{n!}{\prod_{j=1}^{n}w(e_{j})}} \, .
\end{equation}

\end{lemma}

\begin{proof}
 We argue by induction. Clearly the lemma is true for trees with one edge. Suppose the lemma is true for all rooted trees with at most $n-1$ edges. Let $T$ be a rooted tree with $n$ edges. If there is exactly one edge coming out of the root of the tree, say $e_{1}$, then $e_{1}$ must be chosen first in the construction. Therefore, the number of ways of constructing $T$ is equal to the number of ways of constructing the reduced rooted tree formed by removing $e_{1}$ and making the remaining vertex of $e_{1}$ the new root. But,

\begin{equation*}
\begin{split}
 \displaystyle{\frac{n!}{\prod_{j=1}^{n}w(e_{j})}}& =\displaystyle{\frac{n(n-1)!}{w(e_{1})\prod_{j=2}^{n}w(e_{j})}}\\
	& =\displaystyle{\frac{n(n-1)!}{n\prod_{j=2}^{n}w(e_{j})}}\\
	& =\displaystyle{\frac{(n-1)!}{\prod_{j=2}^{n}w(e_{j})}} \, .
\end{split}
\end{equation*}

By the inductive step, this last expression is equal to the number of ways of constructing the reduced rooted tree.\\

Now suppose there are $m$ edges coming out of the root. Each of these edges and their respective descendants form a subtree rooted at the same vertex as $T$. Label these subtrees $T_{1},\ldots,T_{m}$. Each construction of $T$ comes from choosing a construction for each of $T_{1},\ldots,T_{m}$ and then interlacing these sequences (preserving, of course, the ordering of the edges in each subtree) to form a sequence of length $n$ containing all of the edges in $T$. Let subtree $T_{j}$ contain $n_{j} \geq 1$ edges. Then the number of ways to choose the $n_{1}$ locations for the edges in subtree $T_{1}$, the $n_{2}$ locations for the edges in subtree $T_{2}$, etc. is

\begin{equation}
 \binom{n}{n_{1}}\binom{n-n_{1}}{n_{2}}\dots\binom{n-(n_{1}+\ldots+n_{m-1}}{n_{m}} \, .
\end{equation}

Further, by induction, the number of ways to order the edges in subtree $T_{j}$ is

\begin{equation}
 \displaystyle{\frac{n_{j}!}{\prod_{k_{j}=1}^{n_{j}}w(e_{k_{j}}^{j})}} \, .
\end{equation}

Therefore, the number of ways of constructing $T$ is

\begin{equation}
 \displaystyle{\frac{n_{1}!}{\prod_{k_{1}=1}^{n_{1}}w(e_{k_{1}}^{1})}\dots\frac{n_{m}!}{\prod_{k_{m}=1}^{n_{m}}w(e_{k_{m}}^{m})}\cdot\binom{n}{n_{1}}\binom{n-n_{1}}{n_{2}}\dots\binom{n-(n_{1}+\ldots+n_{m-1}}{n_{m}}} \, .
\end{equation}

This is expression is equal to \eqref{treeorders}. This proves the lemma.
\end{proof}

With this lemma in hand, we are ready to construct the sequence of trees that will prove the theorem.

\begin{proof}[Proof of Theorem]

 We begin by considering a lattice tree that looks like that shown in Figure~\ref{tree1}. In that figure, the tree should be thought of as having approximately $n$ total edges, with each vertical segment consisting of approximately $(\log n)^{2}$ edges and with approximately $4(\log (\log n)^{2})^{2}$ edges between the bases of consecutive vertical segments. (Througout this section $\log x$ means $\log_{4} x$.) Then, the length of the long horizontal segment, $f(n)$, must (approximately) satisfy the following equations.

\begin{figure}
 \begin{center}

\includegraphics{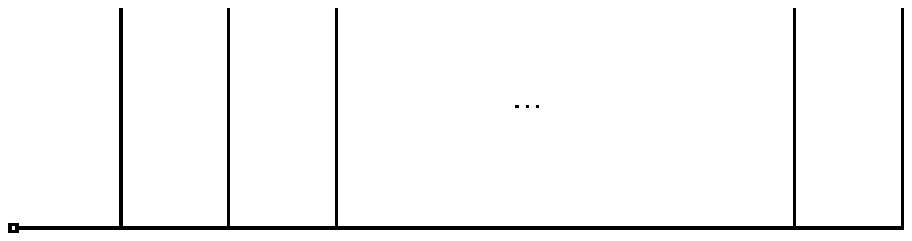}
\caption{}
\label{tree1}
\end{center}
\end{figure}

\begin{gather}\label{fn}
 \displaystyle{\left[ \frac{f(n)}{4(\log (\log n)^{2})^{2}}\right] (\log n)^{2} + f(n) = n}\\
\displaystyle{f(n)\cdot \left[ \frac{(\log n)^{2}}{4(\log (\log n)^{2})^{2}}+1\right] =n}\\
\displaystyle{f(n)\cdot \left[ \frac{(\log n)^{2}+4(\log (\log n)^{2})^{2}}{4(\log (\log n)^{2})^{2}}\right] =n}\\
\displaystyle{f(n)=\frac{4n(\log (\log n)^{2})^{2}}{(\log n)^{2}+4(\log (\log n)^{2})^{2}} \leq \frac{4n(\log (\log n)^{2})^{2}}{(\log n)^{2}}}
\end{gather}

where $\log(\log n)^{2} > 0$.\\

With this calculation in mind, for some large $n$, let each vertical segment consist of $(\log n)^{2}$ edges, let the long horizontal segment consist of $\displaystyle{\frac{4n(\log (\log n)^{2})^{2}}{(\log n)^{2}}=:l(n)}$ edges, and let the number of edges between the bases of consecutive vertical segments be $4(\log (\log n)^{2})^{2}$. (We will choose $n$ such that $l(n)$ is an integer.) Then the total number of edges in the lattice tree is

\begin{equation}
 n + \displaystyle{\frac{4n(\log (\log n)^{2})^{2}}{(\log n)^{2}}} \approx n.
\end{equation}
\\

Now we calculate the product of the weights of the edges. The product of the weights of the edges along all of the vertical segments is

\begin{equation}\label{vertbound1}
 \displaystyle{\left[ ((\log n)^{2})! \right]^{\frac{n}{(\log n)^{2}}}} \, .
\end{equation}

The product of the weights of the edges along the long horizontal segment is

\begin{multline}\label{spineproduct}
 (n+l(n)) \cdot (n+ l(n) -1) \cdot (n+l(n)-2) \cdot \ldots \cdot [n+l(n)-4(\log (\log n)^{2})^{2}+1]\\
\cdot [n+l(n)-4(\log (\log n)^{2})^{2}-(\log n)^{2}] \cdot [n+l(n)-4(\log (\log n)^{2})^{2}-(\log n)^{2}-1] \cdot \ldots \\
\cdot [n+l(n)-8(\log (\log n)^{2})^{2}-(\log n)^{2}+1]\\
\cdot [n+l(n)-8(\log (\log n)^{2})^{2}-2(\log n)^{2}] \cdot \ldots \cdot [n+l(n)-12(\log (\log n)^{2})^{2}-2(\log n)^{2}+1]\\
 \ldots \\
\cdot [4(\log (\log n)^{2})^{2}+(\log n)^{2}] \cdot \ldots \cdot [(\log n)^{2}+1] \, .
\end{multline}

For $n$ large enough, say $n > 1000$, this expression can be bounded above by

\begin{equation}\label{spinebound1}
\displaystyle{\left[ \left(\frac{n}{(\log n)^{2} }\right)! \cdot \left[ 2(\log n)^{2} \right]^{\frac{n}{(\log n)^{2} }} \right]^{4(\log (\log n)^{2})^{2}}} \, .
\end{equation}

We note the following estimate for $N!$.

\begin{equation}\label{factest}
 N! < 2\sqrt{2\pi N}N^{N}e^{-N} < 6\sqrt{N}N^{N}e^{-N}
\end{equation}

Using \eqref{factest}, \eqref{spinebound1} can be bounded by

\begin{equation}
\Bigg[ 6 \left( \frac{n}{(\log n)^{2}} \right) ^{\frac{1}{2}} \cdot \left( \frac{n}{(\log n)^{2}} \right) ^{\frac{n}{(\log n)^{2}}} \cdot e^{\frac{-n}{(\log n)^{2}}} \cdot \left[2(\log n)^{2}\right]^{\frac{n}{(\log n)^{2}}} \Bigg]^{4(\log (\log n)^{2})^{2}}
\end{equation}

\begin{equation}\label{spinebound2}
\begin{split}
& = \left[ 6[g(n)]^{\frac{1}{2}}e^{-g(n)}(2n)^{g(n)} \right]^{4(\log (\log n)^{2})^{2}}\\
& = 6^{4(\log (\log n)^{2})^{2}} \left[ g(n) \right]^{2\left( \log (\log n)^{2} \right)^{2}}e^{-l(n)}(2n)^{l(n)}
\end{split}
\end{equation}

where

\begin{equation}
g(n) \coloneqq \frac{n}{(\log n)^{2}} \, .
\end{equation}

Next note that

\begin{equation}
\displaystyle{\left[ g(n) \right]^{2\left( \log (\log n)^{2} \right)^{2}} < n^{2\left( \log (\log n)^{2} \right)^{2}} = e^{2\left(\ln n \right)\left( \log (\log n)^{2} \right)^{2}} < e^{\frac{n}{(\log n)^{2}}}}
\end{equation}
where the last inequality holds for, say, $\displaystyle{n > 4^{8}}$.\\

Thus, \eqref{spinebound1} can be bounded by

\begin{equation}
\displaystyle{6^{4(\log (\log n)^{2})^{2}} e^{\frac{n}{(\log n)^{2}}-l(n)}(2n)^{l(n)}} \, .
\end{equation}
For $\displaystyle{n > 4^{8}}$, we clearly have $\displaystyle{\frac{n}{(\log n)^{2}} < l(n)}$ \hspace{0.1 in} and \hspace{0.1 in} $\displaystyle{6^{4(\log (\log n)^{2})^{2}} < e^{\frac{4n(\ln 4) \left( \log (\log n)^{2} \right)^{2}}{\log n}}=n^{l(n)}}$. Thus \eqref{spinebound1} is bounded by 

\begin{equation}\label{spinebound3}
 \displaystyle{e^{\frac{12n(\ln 4) \left( \log (\log n)^{2} \right)^{2}}{\log n}}} \, .
\end{equation}

Next we continue our upper bound estimate \eqref{vertbound1} for the vertical edge weights. \eqref{factest} gives

\begin{equation}
\displaystyle{\left[ ((\log n)^{2})! \right]^{\frac{n}{(\log n)^{2}}} < 6^{\frac{n}{(\log n)^{2}}} \cdot (\log n)^{\frac{n}{(\log n)^{2}}} \cdot (\log n)^{2n} \cdot e^{-n}}
\end{equation}

Putting this together with \eqref{spinebound3} we obtain the following upper bound on the product of all the edge weights when $n$ is large. We will use $\displaystyle{n \geq 4^{4^{10}}}$:

\begin{equation}
\begin{split}
& \displaystyle{6^{\frac{n}{(\log n)^{2}}} \cdot (\log n)^{\frac{n}{(\log n)^{2}}} \cdot (\log n)^{2n} \cdot e^{-n} \cdot e^{\frac{12n(\ln 4) \left( \log (\log n)^{2} \right)^{2}}{\log n}}} \\
& = \displaystyle{6^{\frac{n}{(\log n)^{2}}} \cdot (\log n)^{\frac{n}{(\log n)^{2}}} \cdot (\log n)^{2n} \cdot \left[ e^{\frac{12(\ln 4) \left( \log (\log n)^{2} \right)^{2}}{\log n} - 1} \right] ^{n}} \\
& \leq \displaystyle{(6\log n)^{\frac{n}{(\log n)^{2}}} \cdot (\log n)^{2n} \cdot e^{-\frac{24}{25}n}} \\
& \leq \displaystyle{e^{-\frac{1}{2}n}(\log n)^{2n}}
\end{split}
\end{equation}

\begin{figure}
 \begin{center}

\includegraphics{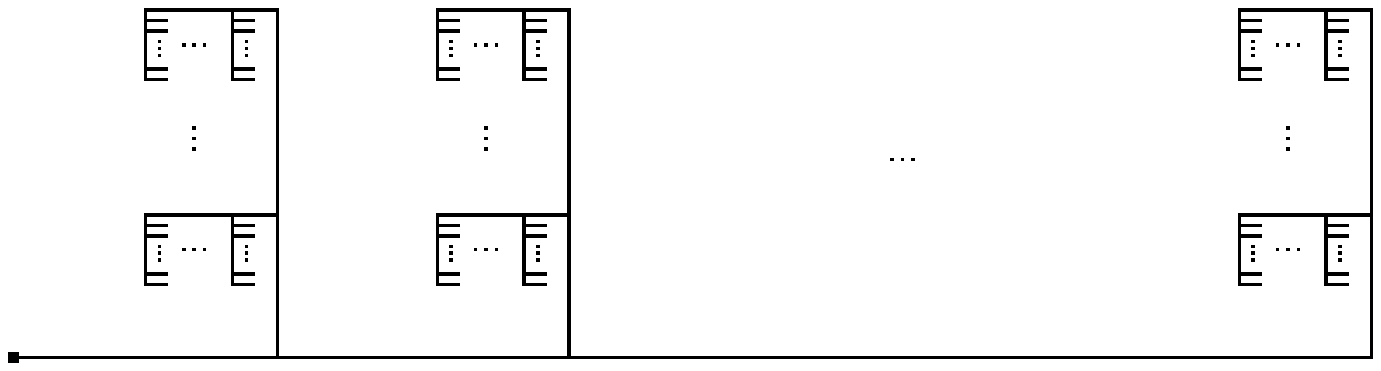}
\caption{}
\label{tree2}
\end{center}
\end{figure}

Now we would like to generalize somewhat the lattice tree shown in Figure~\ref{tree1}. Consider the lattice tree shown in Figure~\ref{tree2}. As in Figure~\ref{tree1}, horizontal and vertical line segments represent many edges that lie end-to-end. Define

\begin{equation}
L(x) \coloneqq (\log_{4}x)^{2}, \hspace{0.25 in} L_{k}(x) \coloneqq \underbrace{L \circ \ldots \circ L}_{k \hspace{0.05 in} \text{times}}(x) \, .
\end{equation}

Further, define
\begin{align}
E(x)& \coloneqq 4^{x}\\
E_{0}& \coloneqq 400\\
E_{1}& \coloneqq 4^{20}\\
E_{k}& \coloneqq E(E_{k-1}^{\frac{1}{2}}) \hspace{0.2 in} \text{for} \hspace{0.05 in} k \geq 2 \, .
\end{align}

Then

\begin{equation}
L_{j}(E_{k}) = E_{k-j} \hspace{0.2 in} \text{for} \hspace{0.05 in} j < k \, .
\end{equation}

Suppose the lattice tree shown in Figure~\ref{tree2} has $n_{k-1} \coloneqq \displaystyle{E_{k} \left[ 1+4 \left( \frac{E_{0}}{E_{1}}+\frac{E_{1}}{E_{2}}+...+\frac{E_{k-2}}{E_{k-1}} \right) \right]}$ edges, and that the figure is composed from horizontal and vertical segments of $k$ different lengths. That is, the shortest path from the root to any leaf will make $k-1$ "left turns". The total number of edges in each cluster attached to the longest horizontal segment is $n_{k-2}$, in each cluster attached to the longest vertical segment is $n_{k-3}$, etc. until the final clusters, which look just like the lattice trees shown in Figure~\ref{tree1} and have $\displaystyle{n_{1}=E_{2}+\frac{4E_{2}E_{0}}{E_{1}}}$ edges. The number of edges in the horizontal and vertical line segments, beginning with the shortest, is given by

\begin{align}
l_{1}& = E_{1}\\
l_{j}& =\displaystyle{\frac{4E_{j}E_{j-2}}{E_{j-1}}} \hspace{0.1 in} \text{for} \hspace{0.05 in} 2 \leq j \leq k \, .
\end{align}

In addition, along a segment of length $l_{j}$, consecutive segments of length $l(j-1)$ branching off of it are separated at their bases by $4E_{j-2}$ edges. We refer to a segment of length $l_{j}$ as a $j^{\text{th}}$\textit{-level segment} and a $j^{\text{th}}$-level segment with all of its "descendants" as a $(j-1)^{\text{th}}$\textit{-level cluster}.\\

Now we calculate an upper bound for the product of the weights of the edges. Using the result of our calculation for the lattice tree in Figure~\ref{tree1} with $n=E_{2}$, we find that the product of the weights of the edges in a $1^{\text{st}}$-level cluster, $W_{2}$, is given by

\begin{equation}
\displaystyle{W_{2}=e^{-\frac{1}{2}E_{2}}E_{1}^{E_{2}}} \, .
\end{equation}

Along each $3^{\text{rd}}$-level segment we have $\frac{E_{3}}{E_{2}}$ $1^{\text{st}}$-level clusters. The product of the weights of the edges in a $3^{\text{rd}}$-level segment can be computed in an almost identical fashion as was done in \eqref{spineproduct} and \eqref{spinebound1}. A simple calculation shows that \eqref{spinebound3} can be used with $n=E_{3}$, and we find that the product of the weights of the edges on each $2^{\text{nd}}$-level cluster is bounded above by

\begin{equation}
\displaystyle{\left[ e^{-\frac{1}{2}E_{2}}E_{1}^{E_{2}} \right] ^{\frac{E_{3}}{E_{2}}} \cdot \left[ E\left(\frac{12E_{1}}{\sqrt{E_{2}}}\right)\right]^{E_{3}}} \, .
\end{equation}

If we continue in this way we obtain

\begin{equation}
\displaystyle{\left[ \ldots \left[ \left[ \left[ e^{-\frac{1}{2}E_{2}}E_{1}^{E_{2}} \right] ^{\frac{E_{3}}{E_{2}}} \cdot \left[ E\left(\frac{12E_{1}}{\sqrt{E_{2}}}\right)\right]^{E_{3}} \right] ^{\frac{E_{4}}{E_{3}}} \cdot \left[ E \left(\frac{12E_{2}}{\sqrt{E_{3}}}\right)\right]^{E_{4}}\right]^{\frac{E_{5}}{E_{4}}} \cdot \ldots \right]^ {\frac{E_{k}}{E_{k-1}}} \cdot \left[ E\left(\frac{12E_{k-2}}{\sqrt{E_{k-1}}}\right)\right]^{E_{k}}}
\end{equation}

which simplifies to

\begin{equation}
e^{-\frac{1}{2}E_{k}} \cdot E_{1}^{E_{k}} \cdot \left[ E\left(\frac{12E_{1}}{\sqrt{E_{2}}}\right)\right]^{E_{k}} \cdot \ldots \cdot \left[ E\left(\frac{12E_{k-2}}{\sqrt{E_{k-1}}}\right)\right]^{E_{k}}
\end{equation}

\begin{equation}
= e^{-\frac{1}{2}E_{k}} \cdot E_{1}^{E_{k}} \left[ E \left(12 \left( \frac{E_{1}}{\sqrt{E_{2}}} + \ldots + \frac{E_{k-2}}{\sqrt{E_{k-1}}} \right) \right) \right]^{E_{k}} \, .
\end{equation}

The infinite series $\displaystyle{\frac{E_{1}}{\sqrt{E_{2}}} + \ldots + \frac{E_{k-2}}{\sqrt{E_{k-1}}} + \ldots}$ converges extremely quickly to something easily seen to be less than $2\frac{E_{1}}{\sqrt{E_{2}}}$. Therefore, we have an upper bound on the product of the weights of the edges given by

\begin{equation}
e^{-\frac{1}{4}E_{k}}(4^{20})^{E_{k}} < (4^{20})^{E_{k}} < (4^{20})^{n_{k-1}} \, .
\end{equation}

Applying the lemma, the theorem follows.
\end{proof}

\ref{Treethm} demonstrates that the number of commutator sequences grows \textit{too fast} to establish a locality result by summing the norms of the ``iterated commutators" in the Taylor expansion of the time-evolved observable. In the real-time case, where a strong locality result is already known, a significant amount of cancellation must occur when these iterated commutators are summed. The cancellation, in fact, occurs for \textit{every} choice of a local interaction and local observable. Should we, then, expect to find a nearest-neighbor interaction for which it can be \textit{shown} that the iterated commutators add \textit{constructively} in the complex-time case, and therefore violate complex-time locality? That this is possible came as a surprise to the author.

\section{Quantum Spin Systems Which Violate Complex-Time Locality}
We begin with several observations concerning operators on $\displaystyle{\otimes_{j=1}^{N} \mathbb{C}^{2}}$ which are tensor products of Pauli spin matrices. First, recall the definition of the Pauli spin matrices:

\begin{equation}
\sigma_{1}=
\begin{pmatrix}
0& 1\\
1& 0
\end{pmatrix}
\hspace{0.4 in}
\sigma_{2}=
\begin{pmatrix}
0& -i\\
i& 0
\end{pmatrix}
\hspace{0.4 in}
\sigma_{3} =
\begin{pmatrix}
1& 0\\
0& -1
\end{pmatrix}
\, .
\end{equation}
Each of the Pauli spin matrices is self-adjoint, unitary and satisfies the commutation relations neatly summarized by

\begin{equation}
 \left[ \sigma_{a},\sigma_{b} \right] = 2i\epsilon_{abc}\sigma_{c}
\end{equation}
where $\epsilon_{abc}$ is the \textit{Levi-Civita symbol}. To simplify the notation in our counterexample we will make the following definitions.

\begin{equation}
\alpha_{0} \coloneqq \mathbbm{1}
\hspace{0.4 in}
\alpha_{1} \coloneqq \sigma_{3}=
\begin{pmatrix}
1& 0\\
0& -1
\end{pmatrix}
\hspace{0.4 in}
\alpha_{2} \coloneqq \sigma_{1}=
\begin{pmatrix}
0& 1\\
1& 0
\end{pmatrix}
\hspace{0.4 in}
\alpha_{3} \coloneqq -i\sigma_{2}=
\begin{pmatrix}
0& 1\\
-1& 0
\end{pmatrix}
\end{equation}
With these definitions we have the following identities:

\begin{equation}\label{alpharel}
\alpha_{1}\alpha_{2}=-\alpha_{2}\alpha_{1}=\alpha_{3}, \hspace{0.3 in} \alpha_{1}\alpha_{3}=-\alpha_{3}\alpha_{1}=\alpha_{2}, \hspace{0.3 in} \alpha_{3}\alpha_{2}=-\alpha_{2}\alpha_{3}=\alpha_{1}
\end{equation} 
and

\begin{equation}\label{alpharel2}
\alpha_{1}^{2}=\mathbbm{1}, \hspace{0.3 in} \alpha_{2}^{2}=\mathbbm{1}, \hspace{0.3 in} \alpha_{3}^{2}=-\mathbbm{1} \, .
\end{equation}

Let $\Lambda \subset \mathbb{Z}^{2}$ be a large square subset containing the origin. Suppose $|\Lambda|=N$. Then an operator on $\displaystyle{\otimes_{x \in \Lambda} \mathbb{C}^{2}}$ of the form $\displaystyle{\otimes_{x \in \Lambda} \alpha_{f(x)}}$, where $f(x) \in \{ 0,1,2,3 \}$, has \textit{Hilbert-Schmidt norm}

\begin{equation}
\displaystyle{  \left\lVert  \otimes_{x \in \Lambda} \alpha_{f(x)}   \right\rVert  _{HS} = (2^{N})^{\frac{1}{2}}} \, .
\end{equation}
Let $\displaystyle{A=\sum_{\{ f: \Lambda \rightarrow \{ 0,1,2,3 \} \, \} } c^{f} \, \otimes_{x \in \Lambda} \alpha_{f(x)}}$, where $c^{f} \in \mathbb{C}$. Then

\begin{equation}
   \left\lVert  A   \right\rVert   \geq \frac{  \left\lVert  A   \right\rVert  _{HS}}{  \left\lVert  \mathbbm{1}   \right\rVert  _{HS}} \geq |c^{f}|
\end{equation}
for all $f:\Lambda \rightarrow \{ 0,1,2,3 \}$.\\

Consider the $2$-dimensional quantum spin system with nearest-neighbor interactions given by $H_{(a,b),(a+1,b)}=\alpha_{1} \otimes \alpha_{2}$ and $H_{(a,b),(a,b-1)}=\alpha_{1} \otimes \alpha_{2}$ for all $(a,b)$. Given a sufficiently large $\Lambda \subset \mathbb{Z}^{2}$ containing the origin and $A=\alpha_{2}$ an operator on $\otimes_{x \in \Lambda} \mathbb{C}^{2}$ supported at the origin, we want to investigate $\mathcal{C}_{\Lambda}^{n}(A)$. To this end, it will prove useful to consider also the $2$-dimensional lattice tree shown in Figure~\ref{tree3}. This lattice tree is an "unfolded" version of the lattice tree in Figure~\ref{tree2}. In this lattice tree, as one moves away from the root (which we place at the origin), one is always moving either right or down. \\

\begin{figure}
 \begin{center}

\includegraphics{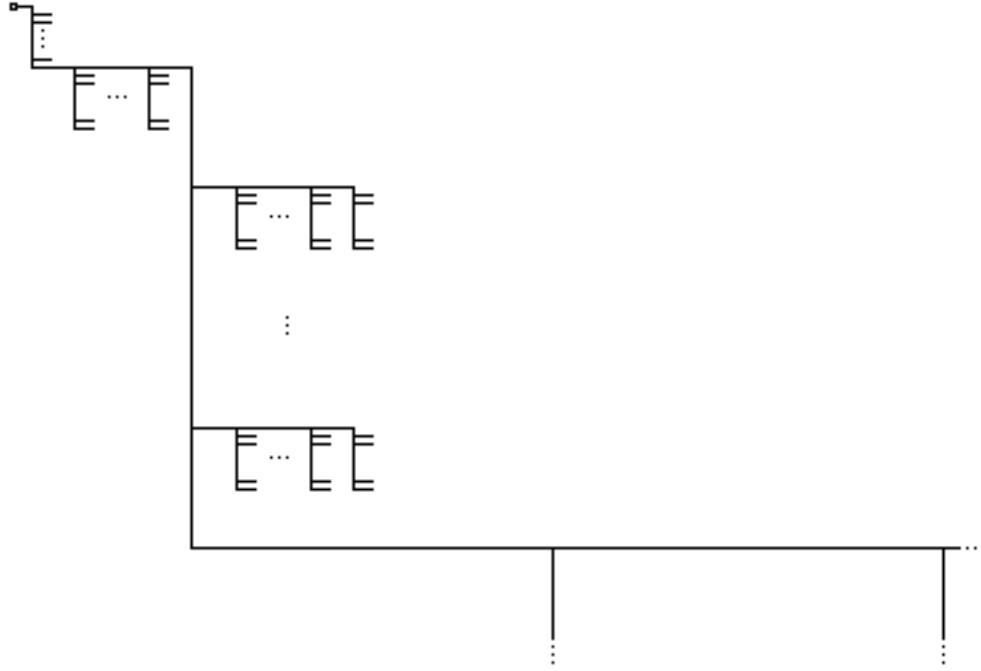}
\caption{The Unfolded Tree}
\label{tree3}
\end{center}
\end{figure}

By slightly modifying the proof of the theorem in the previous section, one can show that the tree of size $n_{j}$, which we will designate $\tilde{T_{j}}$, can be constructed in at least $\frac{n_{j}!}{(4^{21})^{n_{j}}}$ ways. Further, in all of these constructions, the operator associated with the ordered sequence of edges (or lattice animal history) is $B_{j}$, uniquely defined by the following properties:

\begin{enumerate}
 \item It is of the form $2^{n_{j}} \otimes_{x \in \Lambda} \alpha_{f_{j}(x)}$
 \item $f_{j}((0,0))= 3$
 \item $f_{j}(x)= 3$ if $x \in \tilde{T_{j}}$ and degree$(x)=2$
 \item $f_{j}(x)= 2$ if $x \in \tilde{T_{j}}$ and degree$(x)=3$
 \item $f_{j}(x)= 2$ if $x \in \tilde{T_{j}}$, degree$(x)=1$, and $x \neq (0,0)$
 \item $f_{j}(x)= 0$ if $x \notin \tilde{T_{j}}$ \, .
\end{enumerate}

\begin{figure}
\begin{center}
$
\hspace{1.5cm}
 \begin{psmatrix}[emnode=dot,colsep=1.5cm,rowsep=1.2cm]
	& & & &\\
	& & & &\\
	& & & &
  \everypsbox{\scriptstyle}
  \ncline[linewidth=.5pt,linestyle=dotted]{1,1}{1,2}
  \ncline[linewidth=.5pt,linestyle=dotted]{1,2}{1,3}
  \ncline[linewidth=.5pt,linestyle=dotted]{1,3}{1,4}
  \ncline[linewidth=.5pt,linestyle=dotted]{1,1}{2,1}
  \ncline[linewidth=.5pt]{1,2}{2,2}<{b}
  \ncline[linewidth=.5pt]{1,3}{2,3}>{e}
  \ncline[linewidth=.5pt,linestyle=dotted]{1,4}{2,4}
  \ncline[linewidth=.5pt]{2,1}{2,2}_{c}
  \ncline[linewidth=.5pt]{2,2}{2,3}^{a}
  \ncline[linewidth=.5pt]{2,3}{2,4}^{f}
  \ncline[linewidth=.5pt,linestyle=dotted]{2,1}{3,1}
  \ncline[linewidth=.5pt]{2,2}{3,2}>{d}
  \ncline[linewidth=.5pt]{2,3}{3,3}>{g}
  \ncline[linewidth=.5pt,linestyle=dotted]{2,4}{3,4}
  \ncline[linewidth=.5pt,linestyle=dotted]{3,1}{3,2}
  \ncline[linewidth=.5pt,linestyle=dotted]{3,2}{3,3}
  \ncline[linewidth=.5pt,linestyle=dotted]{3,3}{3,4}
 \end{psmatrix}
$
\caption{}
\label{graph1}
\end{center}
\end{figure}

Suppose that neither of the endpoints of the edge labelled $a$ (call them $x_{1}$ and $x_{2}=x_{1}+(1,0)$) in Figure~\ref{graph1} is the origin. If $(e_{1},e_{2}, \ldots ,e_{n})$ is a commutator sequence in $\Lambda$ for which the associated operator $[H_{n},[H_{n-1}, \ldots ,[H_{2},[H_{1},A]] \ldots ]]$ is not the zero operator, what are the requirements on $(e_{1},e_{2}, \ldots ,e_{n})$ so that the edge labelled $a$ can appear in the the $(n+1)^{th}$ position? (That is, so that the operator associated with $(e_{1},e_{2}, \ldots ,e_{n},a)$ is also not the zero operator?) If the sum of the number of appearances in the commutator sequence of $b$ and $c$ is odd, if $a$, $d$, and $e$ have each appeared an even number of times in the commutator sequence, and if the sum of the number of appearances of $f$ and $g$ is even, then, by \eqref{alpharel} and \eqref{alpharel2}, the associated iterated commutator will have an $\alpha_{2}$ in the $x_{1}-$position and an $\mathbbm{1}$ in the $x_{2}-$position. (This is a simple calculation.) Thus, using \eqref{alpharel} again, the operator associated to $(e_{1},e_{2}, \ldots ,e_{n},a)$ will have an $\alpha_{3}$ in the $x_{1}-$position, an $\alpha_{2}$ in the $x_{2}-$position, all of the other positions will remain unchanged, and the overall coefficient will be multiplied by $2$. Similar calculations can be made for all other parities of the edges (or edge combinations). The results are summarized in Table~\ref{opertable}.\\

\begin{table}[t]\label{table}
 \caption{}
 \centering
\begin{tabular}{|m{1 cm}|m{1 cm}|m{1 cm}|m{1 cm}|m{1 cm}|m{1 cm}|m{1 cm}|m{1 cm}|m{1 cm}|m{1 cm}|}
\hline $a$ & $b$ and $c$ combined & $d$ & $e$ & $f$ and $g$ combined & Oper. in $x_{1}-$pos. & Oper. in $x_{2}-$pos. & Oper. in $x_{1}-$pos. after $a$ & Oper. in $x_{2}-$pos. after $a$ & Coeff. multiplier\\
\hline even & even & even & even & even & $\mathbbm{1}$ & $\mathbbm{1}$ & $0$ & $0$ & $0$\\
\hline odd & even & even & even & even & $\alpha_{1}$ & $\mathbbm{1}$ & $0$ & $0$ & $0$\\
\hline even & odd & even & even & even & $\alpha_{2}$ & $\mathbbm{1}$ & $\alpha_{3}$ & $\alpha_{2}$ & $2$\\
\hline even & even & odd & even & even & $\alpha_{1}$ & $\mathbbm{1}$ & $0$ & $0$ & $0$\\
\hline even & even & even & odd & even & $\mathbbm{1}$ & $\alpha_{2}$ & $0$ & $0$ & $0$\\
\hline even & even & even & even & odd & $\mathbbm{1}$ & $\alpha_{1}$ & $\alpha_{1}$ & $\alpha_{3}$ & $-2$\\
\hline odd & odd & even & even & even & $\alpha_{3}$ & $\alpha_{2}$ & $\alpha_{2}$ & $\mathbbm{1}$ & $2$\\
\hline odd & even & odd & even & even & $\mathbbm{1}$ & $\alpha_{2}$ & $0$ & $0$ & $0$\\
\hline odd & even & even & odd & even & $\alpha_{1}$ & $\mathbbm{1}$ & $0$ & $0$ & $0$\\
\hline odd & even & even & even & odd & $\alpha_{1}$ & $\alpha_{3}$ & $\mathbbm{1}$ & $\alpha_{1}$ & $-2$\\
\hline even & odd & odd & even & even & $\alpha_{3}$ & $\mathbbm{1}$ & $\alpha_{2}$ & $\alpha_{2}$ & $2$\\
\hline even & odd & even & odd & even & $\alpha_{2}$ & $\alpha_{2}$ & $\alpha_{3}$ & $\mathbbm{1}$ & $2$\\
\hline even & odd & even & even & odd & $\alpha_{2}$ & $\alpha_{1}$ & $0$ & $0$ & $0$\\
\hline even & even & odd & odd & even & $\alpha_{1}$ & $\alpha_{2}$ & $0$ & $0$ & $0$\\
\hline even & even & odd & even & odd & $\alpha_{1}$ & $\alpha_{1}$ & $\mathbbm{1}$ & $\alpha_{3}$ & $-2$\\
\hline even & even & even & odd & odd & $\mathbbm{1}$ & $\alpha_{3}$ & $\alpha_{1}$ & $\alpha_{1}$ & $-2$\\
\hline odd & odd & odd & even & even & $\alpha_{2}$ & $\alpha_{2}$ & $\alpha_{3}$ & $\mathbbm{1}$ & $2$\\
\hline odd & odd & even & odd & even & $\alpha_{3}$ & $\mathbbm{1}$ & $\alpha_{2}$ & $\alpha_{2}$ & $2$\\
\hline odd & odd & even & even & odd & $\alpha_{3}$ & $\alpha_{3}$ & $0$ & $0$ & $0$\\
\hline odd & even & odd & odd & even & $\mathbbm{1}$ & $\mathbbm{1}$ & $0$ & $0$ & $0$\\
\hline odd & even & odd & even & odd & $\mathbbm{1}$ & $\alpha_{3}$ & $\alpha_{1}$ & $\alpha_{1}$ & $-2$\\
\hline odd & even & even & odd & odd & $\alpha_{1}$ & $\alpha_{1}$ & $\mathbbm{1}$ & $\alpha_{3}$ & $-2$\\
\hline even & odd & odd & odd & even & $\alpha_{3}$ & $\alpha_{2}$ & $\alpha_{2}$ & $\mathbbm{1}$ & $2$\\
\hline even & odd & odd & even & odd & $\alpha_{3}$ & $\alpha_{1}$ & $0$ & $0$ & $0$\\
\hline even & odd & even & odd & odd & $\alpha_{2}$ & $\alpha_{3}$ & $0$ & $0$ & $0$\\
\hline even & even & odd & odd & odd & $\alpha_{1}$ & $\alpha_{3}$ & $\mathbbm{1}$ & $\alpha_{1}$ & $-2$\\
\hline odd & odd & odd & odd & even & $\alpha_{2}$ & $\mathbbm{1}$ & $\alpha_{3}$ & $\alpha_{2}$ & $2$\\
\hline odd & odd & odd & even & odd & $\alpha_{2}$ & $\alpha_{3}$ & $0$ & $0$ & $0$\\
\hline odd & odd & even & odd & odd & $\alpha_{3}$ & $\alpha_{1}$ & $0$ & $0$ & $0$\\
\hline odd & even & odd & odd & odd & $\mathbbm{1}$ & $\alpha_{1}$ & $\alpha_{1}$ & $\alpha_{3}$ & $-2$\\
\hline even & odd & odd & odd & odd & $\alpha_{3}$ & $\alpha_{3}$ & $0$ & $0$ & $0$\\
\hline odd & odd & odd & odd & odd & $\alpha_{2}$ & $\alpha_{1}$ & $0$ & $0$ & $0$\\
\hline
\end{tabular}
\label{opertable}
\end{table}

Note that the coefficient multiplier only depends on the parities of two columns: $b$ and $c$ combined and $f$ and $g$ combined. If they are the same, the multiplier is $0$. If they are different and $f$ and $g$ combined is even, then the multiplier is $2$. If they are different and $f$ and $g$ combined is odd, then the multiplier is $-2$. A completely analogous analysis can be made for edge $a$ in Figure~\ref{graph2}, and the results are the same as in Table~\ref{opertable}. (In this case, the endpoints of $a$ are $x_{1}$ and $x_{2}=x_{1}+(0, -1)$.)\\

\begin{figure}
\begin{center}
$
\hspace{1.5cm}
 \begin{psmatrix}[emnode=dot,colsep=1.5cm,rowsep=1.2cm]
       & & &\\
       & & &\\
       & & &\\
       & & &	
 \large
  \everypsbox{\scriptstyle}
  \ncline[linewidth=.5pt,linestyle=dotted]{1,1}{1,2}
  \ncline[linewidth=.5pt,linestyle=dotted]{1,2}{1,3}
  \ncline[linewidth=.5pt,linestyle=dotted]{1,1}{2,1}
  \ncline[linewidth=.5pt]{1,2}{2,2}<{c}
  \ncline[linewidth=.5pt,linestyle=dotted]{1,3}{2,3}
  \ncline[linewidth=.5pt]{2,1}{2,2}^{b}
  \ncline[linewidth=.5pt]{2,2}{2,3}^{d}
  \ncline[linewidth=.5pt,linestyle=dotted]{2,1}{3,1}
  \ncline[linewidth=.5pt]{2,2}{3,2}>{a}
  \ncline[linewidth=.5pt,linestyle=dotted]{2,3}{3,3}
  \ncline[linewidth=.5pt]{3,1}{3,2}_{e}
  \ncline[linewidth=.5pt]{3,2}{3,3}_{g}
  \ncline[linewidth=.5pt,linestyle=dotted]{3,1}{4,1}
  \ncline[linewidth=.5pt]{3,2}{4,2}<{f}
  \ncline[linewidth=.5pt,linestyle=dotted]{3,3}{4,3}
  \ncline[linewidth=.5pt,linestyle=dotted]{4,1}{4,2}
  \ncline[linewidth=.5pt,linestyle=dotted]{4,2}{4,3}
 \end{psmatrix}
$
\caption{}
\label{graph2}
\end{center}
\end{figure}

\begin{lemma} If a commutator sequence of length $n$ in $\Lambda$ results in an operator that is a (non-zero) scalar multiple of the operator $B_{j}$, then it is a positive scalar multiple of $B_{j}$ and $n-n_{j}$ is an even integer.
\end{lemma}

\begin{proof}
 We first observe that each edge in the planar square lattice is contained in exactly one ``cross-shaped'' group of four edges that either all share a vertex with even coordinates or all share a vertex with odd coordinates. Let $(e_{1}, \ldots ,e_{n})$ be a commutator sequence that results in a non-zero operator. Now consider the ``cross-shaped'' group of four edges that all contain the origin and label them as shown in Figure~\ref{graph3}. To each appearance of an edge labelled $a$, $b$, $c$, or $d$ (or really any edge, for that matter) is associated a coefficient multiplier, which is either $2$ or $-2$. We would like to show that the product of the coeffiecient multipliers associated to the appearances of the edges labelled $a$, $b$, $c$, or $d$ is positive.\\

\begin{figure}
\begin{center}
$
\hspace{1.5cm}
 \begin{psmatrix}[emnode=dot,colsep=1.5cm,rowsep=1.2cm]
       & & &\\
       & & &\\
       & & &
 \large
  \everypsbox{\scriptstyle}
  \ncline[linewidth=.5pt,linestyle=dotted]{1,1}{1,2}
  \ncline[linewidth=.5pt,linestyle=dotted]{1,2}{1,3}
  \ncline[linewidth=.5pt,linestyle=dotted]{1,1}{2,1}
  \ncline[linewidth=.5pt]{1,2}{2,2}<{b}
  \ncline[linewidth=.5pt,linestyle=dotted]{1,3}{2,3}
  \ncline[linewidth=.5pt]{2,1}{2,2}_{c}
  \ncline[linewidth=.5pt]{2,2}{2,3}^{a}
  \ncline[linewidth=.5pt,linestyle=dotted]{2,1}{3,1}
  \ncline[linewidth=.5pt]{2,2}{3,2}>{d}
  \ncline[linewidth=.5pt,linestyle=dotted]{2,3}{3,3}
  \ncline[linewidth=.5pt,linestyle=dotted]{3,1}{3,2}
  \ncline[linewidth=.5pt,linestyle=dotted]{3,2}{3,3}
 \end{psmatrix}
$
\caption{}
\label{graph3}
\end{center}
\end{figure}

Note that the results collected in the previous table imply that the value of the multiplicative constant associated with the appearance of an edge labelled $a$ or $d$ depends only on the parity of the combined number of previous appearances of $b$ and $c$. (However, because our group of edges is centered around the origin, and because we are commuting with an operator $A$ which is an $\alpha_{2}$ supported at the origin, we must add one to the number of combined appearances of $b$ and $c$ when calculating the value of the multiplicative constant associated with the appearance of an $a$ or $d$.) Therefore, we may restrict ourselves to the subsequence of $(e_{1}, \ldots ,e_{n})$ formed by removing all $e_{k}$ which are not equal to $a$, $b$, $c$, or $d$.\\

Since $B_{j}$ has an $\alpha_{3}$ at the origin, $b$ and $c$ must combined appear an even number of times and $a$ and $d$ must combined appear an odd number of times in the commutator sequence $(e_{1}, \ldots ,e_{n})$ (hence in the subsequence described above). Otherwise, the associated operator will differ from $B_{j}$ at the origin. We claim that out of all of these appearances of $a$, $b$, $c$, and $d$ an even number of them will introduce a negative coefficient multiplier.\\

Our subsequence consists of alternating $b \, c$ strings and $a \, d$ strings. An example might be

\begin{equation*}
 \underbrace{a \, a \, d \, a} \, \underbrace{b \, b \, c} \, \underbrace{a \, a} \, \underbrace{c \, b \, c} \, \underbrace{a \, d \, d \, a \, d \, d} \, \underbrace{c \, c \, c \, b \, b} \, .
\end{equation*}
Suppose our subsequence has an $a \, d$ string of even length. Since the sign of the coefficient multiplier introduced by the appearance of each $a$ or $d$ depends only on the parity of the number of combined previous appearances of $b$ and $c$, each appearance of $a$ or $d$ in this string will introduce the same coefficient multiplier. Since this $a \, d$ string is of even length, the product of the coefficient multipliers associated to the edges in this string will be positive. Further, this string will not affect the coefficient multipliers associated to any other edges in the subsequence. Thus, for analyzing the parity of the product of coefficient multipliers in our $a \, b \, c \, d$ subsequence, we may remove $a \, d$ and $b \, c$ strings of even length. This process may be repeated until either we have no more entries in our subsequence (in which case the product of the coefficient multipliers is positive), or we have a reduced subsequence consisting of alternating strings of odd length.\\

For ease of notation we will represent an $a \, d$ string of odd length by an $a$ and a $b \, c$ string of odd length by a $b$. If an $a \, d$ string is first in our (reduced) subsequence, then the subsequence is of the form

\begin{center}
\begin{tabular}{cccccccccc}
$a$ & $b$ & $a$ & $b$ & $a$ & \dots & $b$ & $a$ & $b$ & $a$
\end{tabular}
\end{center}
where $a$ appears an odd number of times and $b$ appears an even number of times. Because our observable $A$ is an $\alpha_{2}$ supported at the origin, the appearance of an $a$ or $d$ will introduce a positive coefficient multiplier when $b$ and $c$ have combined already appeared an \textit{even} number of times. So, the signs of the coefficient multipliers are\\
\\

\begin{center}
\begin{tabular}{cccccccccc}
$a$ & $b$ & $a$ & $b$ & $a$ & \dots & $b$ & $a$ & $b$ & $a$\\
$+$ & $-$ & $-$ & $+$ & $+$ & \dots & $-$ & $-$ & $+$ & $+$
\end{tabular}
\end{center}

Similarly, if a $b$-$c$ string appears first, then we have

\begin{center}
\begin{tabular}{cccccccccc}
$b$ & $a$ & $b$ & $a$ & $b$ & \dots & $a$ & $b$ & $a$ & $b$\\
$+$ & $-$ & $-$ & $+$ & $+$ & \dots & $+$ & $+$ & $-$ & $-$
\end{tabular}
\end{center}

Thus, negative coefficient multipliers associated with edges labelled $a$, $b$, $c$, or $d$ will combined appear an even number of times. We also note that in any commutator sequence in which the edges of $\tilde{T_{j}}$ each appear exactly once and such that the resulting operator is non-zero (that is, so that we obtain $B_{j}$), edge $a$ appears exactly once and edges $b$, $c$, and $d$ do not appear at all. In the commutator sequence $(e_{1}, \ldots ,e_{n})$ that we are considering, edges $a$, $b$, $c$ and $d$ combined appear an odd number of times. We will need this observation later when we want to show that $n-n{j}$ is even.\\

By construction, in the lattice tree $\tilde{T_{j}}$, degree $1$ and degree $3$ vertices are always an even number of edges away from the next nearest degree $1$ or degree $3$ vertex. Similarly, vertices which are degree $2$ and at which a ``turn'' occurs are always separated from the nearest degree $1$ or degree $3$ vertex by an even number of edges. So, when we cover the planar square lattice with the cross-shapes described above, seven additional cases will occur (with the edges in the group labelled as in Figure~\ref{graph3}).
\begin{enumerate}
 \item $a$ and $c$ are in $\tilde{T_{j}}$, and $b$ and $d$ are not.
 \item $c$ and $d$ are in $\tilde{T_{j}}$, and $c$ and $d$ are not.
 \item $a$ and $b$ are in $\tilde{T_{j}}$, and $c$ and $d$ are not.
 \item $a$, $b$, and $d$ are in $\tilde{T_{j}}$, but $c$ is not.
 \item $a$, $c$, and $d$ are in $\tilde{T_{j}}$, but $b$ is not.
 \item $c$ is in $\tilde{T_{j}}$, but $a$, $b$, and $d$ are not.
 \item $a$, $b$, $c$ and $d$ are not in $\tilde{T_{j}}$.
\end{enumerate}

We can analyze each of these cases in exactly the same way we analyzed the group of four edges that contain the origin. In each case we can first reduce the commutator sequence $(e_{1}, \ldots ,e_{n})$ to the subsequence containing just the edges in the cross-shaped group. Then we can further reduce the subsequence until we have alternating $a$-$d$ and $b$-$c$ strings of odd length.\\

\textit{Case 1}. In this case, the shared vertex has an $\alpha_{3}$ associated with it. Thus, $b$ and $c$ combined must appear an odd number of times, and $a$ and $d$ combined must also appear an odd number of times. Using the same notation as above for $a$-$d$ strings and $b$-$c$ strings, the signs of the coefficient multipliers associated with the appearance of each edge in the (reduced) subsequence are

\begin{center}
\begin{tabular}{cccccccccc}
$a$ & $b$ & $a$ & $b$ & \dots & $a$ & $b$\\
$-$ & $-$ & $+$ & $+$ & \dots & $-$ & $-$
\end{tabular}
\end{center}

\begin{center}
 or
\end{center}

\begin{center}
\begin{tabular}{cccccccccc}
$b$ & $a$ & $b$ & $a$ & \dots & $b$ & $a$\\
$+$ & $+$ & $-$ & $-$ & \dots & $+$ & $+$
\end{tabular}
\end{center}
Thus, negative coefficient multipliers associated with edges labelled $a$, $b$, $c$, or $d$ will combined appear an even number of times. We also note that $a$, $b$, $c$ and $d$ combined appear an even number of times in $e_{1}, \ldots ,e_{n}$.\\
\\

\textit{Case 2}. In this case the shared vertex is an $\alpha_{3}$. The analysis is just like Case 1.\\
\\

\textit{Case 3}. Just like Cases 1 and 2.\\
\\

\textit{Case 4}. In this case the shared vertex is an $\alpha_{2}$. Edges $a$ and $d$ must combined appear an even number of times and $b$ and $c$ must combined appear an odd number of times. The signs of the coefficient multipliers associated with the appearance of each edge in the (reduced) subsequence are

\begin{center}
\begin{tabular}{cccccccccc}
$a$ & $b$ & $a$ & $b$ & $a$ & $b$ & \dots & $a$ & $b$ & $a$\\
$-$ & $-$ & $+$ & $+$ & $-$ & $-$ & \dots & $-$ & $-$ & $+$
\end{tabular}
\end{center}

\begin{center}
 or
\end{center}

\begin{center}
\begin{tabular}{cccccccccc}
$b$ & $a$ & $b$ & $a$ & $b$ & $a$ & \dots & $b$ & $a$ & $b$\\
$+$ & $+$ & $-$ & $-$ & $+$ & $+$ & \dots & $-$ & $-$ & $+$
\end{tabular}
\end{center}
Thus, negative coefficient multipliers associated with edges labelled $a$, $b$, $c$, or $d$ will combined appear an even number of times. We also note that $a$, $b$, $c$ and $d$ combined appear an odd number of times in $e_{1}, \ldots ,e_{n}$.\\
\\

\textit{Case 5}. Completely analogous to Case 4.\\
\\

\textit{Case 6}. In this case the shared vertex is also an $\alpha_{2}$, and the rest of the analysis is just like the previous two cases.\\
\\

\textit{Case 7}. In this case the shared vertex is an $\mathbbm{1}$. Edges $a$ and $d$ must combined appear an even number of times and $b$ and $c$ must also combined appear an even number of times. The signs of the coefficient multipliers associated with the appearance of each edge in the (reduced) subsequence are

\begin{center}
\begin{tabular}{cccccccccc}
$a$ & $b$ & $a$ & $b$ & \dots & $a$ & $b$\\
$-$ & $-$ & $+$ & $+$ & \dots & $-$ & $-$
\end{tabular}
\end{center}

\begin{center}
 or
\end{center}

\begin{center}
\begin{tabular}{cccccccccc}
$b$ & $a$ & $b$ & $a$ & \dots & $b$ & $a$\\
$+$ & $+$ & $-$ & $-$ & \dots & $+$ & $+$
\end{tabular}
\end{center}
Again, negative coefficient multipliers associated with edges labelled $a$, $b$, $c$, or $d$ will combined appear an even number of times, and the total number of appearances of the edges labelled $a$, $b$, $c$, or $d$ will be even. The lemma follows.

\end{proof}

\begin{MT}
 Let an interaction on the planar square lattice be as described above, and let $\{ \Lambda_{j} \}$ be an increasing sequence of simply-connected ``square-shaped'' subsets of $\mathbb{Z}^{2}$ such that $\tilde{T_{j}} \subset \Lambda_{j}$. Also let $A=\alpha_{2}$ supported at the origin. Then,

\begin{equation}
 \displaystyle{\lim_{j \to \infty}   \left\lVert  e^{izH_{\Lambda_{j}}}Ae^{-izH_{\Lambda_{j}}}   \right\rVert   = \infty}
\end{equation}
for $z$ purely imaginary and $|z| > 4^{21}$.
\end{MT}

\begin{proof}
 We begin by writing

\begin{equation}
 e^{izH_{\Lambda_{j}}}Ae^{-izH_{\Lambda_{j}}}=\sum_{n=0}^{\infty}\frac{z^{n}}{n!}\mathcal{C}_{\Lambda}^{n}(A) \, .
\end{equation}
For any $n$, $\mathcal{C}_{\Lambda}^{n}(A)$ can be written as a (finite) sum of non-zero ``iterated commutators''. Each of these non-zero iterated commutators is an operator of the form $\displaystyle{c^{f}\otimes_{x \in \Lambda_{j}} \alpha_{f(x)}}$, for some $f:\Lambda_{j} \rightarrow \{ 0,1,2,3 \}$. In the previous lemma we showed that if an $n-$times iterated commutator results in an operator of the above form with $f=f_{j}$, then $\displaystyle{c^{f}\otimes_{x \in \Lambda_{j}} \alpha_{f(x)}}$ is a positive scalar multiple of $B_{j}$ and $n-n_{j}$ is even. Noting that $\{ \otimes_{x \in \Lambda_{j}} \alpha_{f(x)} | f:\Lambda_{j} \rightarrow \{ 0,1,2,3 \} \}$ is an orthogonal (but not orthonormal!) basis for the vector space of linear operators on $\otimes_{x \in \Lambda_{j}} \mathbb{C}^{2}$ with the Hilbert-Schmidt inner product, we can write

\begin{equation}
 \displaystyle{e^{izH_{\Lambda_{j}}}Ae^{-izH_{\Lambda_{j}}}=\sum_{n=0}^{\infty}\frac{z^{n}}{n!}\mathcal{C}_{\Lambda}^{n}(A)=\sum_{f:\Lambda_{j} \rightarrow \{ 0,1,2,3 \} } a_{\Lambda_{j}}^{f} \otimes_{x \in \Lambda_{j}} \alpha_{f(x)}} \, .
\end{equation}

Then,
\begin{equation}
 |a_{\Lambda_{j}}^{f_{j}}| \geq \frac{|z|^{n_{j}}}{n_{j}!} \cdotp \frac{n_{j}!}{(4^{21})^{n_{j}}} \cdotp 2^{n_{j}}= \left(\frac{2|z|}{4^{21}} \right)^{n_{j}} \, .
\end{equation}

Since $  \left\lVert  e^{izH_{\Lambda_{j}}}Ae^{-izH_{\Lambda_{j}}}   \right\rVert   \geq |a_{\Lambda_{j}}^{f_{j}}|$, the result follows by letting $j$ go to $\infty$.
\end{proof}

\section{Acknowledgements}
I wish to thank my advisor, Eric Carlen, for many helpful discussions during this project. The idea for the sequence of trees that can be constructed in ``too many" ways was suggested by J\'{o}zsef Beck after several conversations about lattice trees. My fellow graduate student Elizabeth Kupin is responsible for the beautiful lemma on the number of ways of ``constructing" a rooted tree. I am also thankful for helpful conversations with David Ruelle and Giovanni Gallavotti.

\bibliographystyle{amsplain}

\end{document}